\documentclass[runningheads]{llncs}
\usepackage{times}
\usepackage{helvet}
\usepackage{courier}
\usepackage{url}
\usepackage{graphicx}

\usepackage{proof}

\newtheorem{Define}[proposition]{Definition}

\begin{document}
\bibliographystyle{alpha}
\title{{\bf Properties and Extensions of Alternating Path Relevance - I}}
\author{David A. Plaisted
}
\authorrunning{D. Plaisted}
\institute{UNC Chapel Hill, Chapel Hill, NC 27599-3175, USA\\
\email{plaisted@cs.unc.edu}}
\maketitle
\begin{abstract}
When proving theorems from large sets of logical assertions,
it can be helpful to restrict the search for a proof to those assertons that are {\em relevant}, that is, closely related
to the theorem in some sense.  For example, in the Watson system, a large knowledge base must rapidly be searched for relevant facts. It is possible to define formal concepts of relevance for propositional and first-order logic.
Various concepts of relevance have been defined for this, and some
have yielded good results on large problems.  We consider here in particular a concept based on alternating paths.  We
present efficient graph-based methods for computing alternating path relevance and give some results indicating its effectiveness.  We also propose an alternating path based
extension of this relevance method to DPLL with an improved time bound, and give other extensions to alternating path relevance
intended to improve its performance.

\end{abstract}

\section{Introduction}
In some applications, there may be knowledge bases containing thousands or even millions of assertions.  Selecting relevant assertions has the potential to significantly reduce the cost of finding a proof of a desired conclusion from such large knowledge bases and even from smaller ones.  Relevance can be defined in many ways.  In first-order clausal theorem proving, a concept of relevance based on "alternating" paths between clauses \cite{DBLP:journals/ai/PlaistedY03} permits relevant facts to be chosen automatically.
Here we present efficient graph-based methods for computing alternating path relevance.  We give some theoretical and
practical results indicating its effectiveness.  We also present a relationship between alternating path relevance and the set of
support strategy \cite{woroca:65} in theorem proving. We incorporate this approach to relevance into the DPLL method.  Finally, we
present some extensions to alternating path relevance with a view to improving its effectiveness.

\subsection{Related Work}
Meng and Paulson \cite{DBLP:journals/japll/MengP09} describe a relevance approach in which clauses are relevant if they share enough symbols with clauses that have previously been found to be relevant.  They give clauses a “pass mark,” a number between 0 and 1.  This is used as a filter, and the tests becomes increasingly strict as the distance from the conjecture increases.  Their method makes use of two parameters, $p$ and $c$.  They finally chose $p = 0.6$ and $c = 2.4$ based on experiments. This approach significantly improves the performance of the Sledgehammer theorem
prover \cite{DBLP:journals/iandc/MengQP06} used with Isabelle \cite{DBLP:conf/tphol/WenzelPN08}.  They also
tried the alternating path relevance approach \cite{DBLP:journals/ai/PlaistedY03}, but apparently did not set a bound on relevance distance, but rather included all of
the relevant clauses.

Pudlak \cite{DBLP:conf/cade/Pudlak07} defines relevance in terms of finite models.  The idea is to find a set ${\bf B}$ of clauses such that all of the finite models
that have been constructed so far and satisfy ${\bf B}$, also
satisfy the theorem $C$.  Such a set ${\bf B}$ of clauses is a candidate for a sufficient set for proving the theorem.  Clauses $F$ are added to
${\bf B}$ by constructing models of the theorem that do not satisfy ${\bf B}$, and finding clauses $F$ that contradict such models.
This approach is attractive because humans seem to use semantics when proving theorem.  It
had good results on the bushy division of MPTP \cite{DBLP:journals/jar/Urban04}, and was extended in the SRASS system \cite{DBLP:conf/cade/SutcliffeP07}
where it was among the
most successful systems in the MPTP division.  The SRASS system uses a syntactic similarity measure in addition to semantics, as an aid in ordering the axioms. However, the basic system \cite{DBLP:conf/cade/Pudlak07} appears to be inefficient in the presence of a large number of clauses.  SRASS is also resource intensive for large theories.

The MaLARea system \cite{DBLP:conf/cade/Urban07} uses machine learning from previous proof attempts to guide the proof of new problems.  It uses a Bayesian learning system to choose a relevance ordering of the axioms based on the symbols that appear in the conjecture.   In
each proof attempt, the $k$ most relevant clauses are used for various $k$, and various time limits are tried.  It has given good performance on the "chainy" division of MPTP, on which it solved more problems than E, SPASS, and SRASS.  This system has been modified \cite{DBLP:conf/cade/UrbanSPV08} to take into account semantic features of the axioms
in choosing the relevance ordering, similar to SRASS, and the combined system has outperformed both MaLARea and SRASS on the MPTP problems.  The basic MaLARea system has been combined with neural network learning and has performed very well on the large theory batch division in the CASC competitions \cite{Sut16} in 2017 and 2018.  The Divvy system \cite{DBLP:conf/cade/RoedererPS09} also orders the axioms, but does so
solely on the basis of a syntactic similarity measure.  For each proof attempt, a subset of the most relevant axioms is used.  This system has obtained good results on the MPTP challenge problems.

The Sine Qua Non system \cite{DBLP:conf/cade/HoderV11} evaluates how closely two clauses are related by how many symbols they share, and which
symbols they share.  It also considers the length of paths between clauses as a measure of how closely they are related.  The d-relevance idea of Sine Qua Non is straightforward, but in order to reduce the number of relevant axioms, a triggering technique is used.  However, using the least common symbol as a trigger doesn’t work well, so they add a tolerance, a depth limit, and a generality threshold.  Experiments revealed that some of these parameters don’t make much of a difference, and at least the depth limit turned out to be important.  Tolerance was less important for the problems tried.  This relevance measure can be computed very fast even for large sets of clauses.  This approach has performed very well in
the large theory division of CASC, and has been used by Vampire \cite{10.1007/3-540-48660-7_26} and by some competing systems as well.  However, sometimes
clauses may be closely related in the Sine Qua Non system but not in the alternating path approach; for example, the literals
$P(a,b)$ and $\neg P(b,a)$ are closely related in the former system but unrelated in the latter.

These systems do at least show that relevance can be effective in aiding theorem provers on large knowledge bases.  Also, especially relevance measures that can be computed quickly may lead to spectacular increases in the effectiveness of theorem provers on very large knowledge bases.  Even if the original knowledge base is small, large numbers of clauses may be generated during a
proof attempt, and relevance techniques may help to select derived clauses leading to a proof.
\section{Terminology}

\subsection{Connectedness of two clauses}
This section contains basic definitions  related to alternating paths.  Some of the formalism is new.  Standard definitions of terms, clauses, sets of clauses, substitutions, resolution, resolution proof, and satisfiability in first-order logic are assumed.  If $A$ is an atom, then
$A$ and $\neg A$ are literals, and each is the complement of the other.  If $L$ is $\neg A$ then
$\neg L$ will generally denote $A$.

The number of clauses in a set $S$ of first-order clauses is denoted by $|S|$.

The relation $\equiv$ on literals denotes syntactic identity.  That is, it is the smallest relation
on literals such that for all atoms $A$, $A \equiv A$, $\neg A \equiv \neg A$, $A \equiv \neg \neg A$, and
$\neg \neg A \equiv A$.  This treatment of negation permits one to say that $A$ and $B$
are complementary literals iff $A \equiv \neg B$.

A pair $L$ and $M$ of literals are {\em complementary unifiable} if there are substitutions
$\alpha$ and $\beta$ such that $L \alpha \equiv \neg M\beta$.

An {\em alternating path} from $C_1$ to $C_n$ in a set $S$ of clauses is a sequence $C_1, p_1, C_2, p_2$, $\dots,
C_{n-1}, p_{n-1}, C_n$ where $C_i \in S$ for all $i$, $p_i$ is a pair
$(L_i,M_{i+1})$ of literals, $L_i \in C_i$, $M_{i+1} \in C_{i+1}$, $L_i$ and $M_{i+1}$ are complementary unifiable, and $L_i \not\equiv M_i$ for all $i$.  Frequently the $p_i$ are omitted.  Such a path is called {\em alternating} because it alternates between connecting literals in possibly different clauses and switching to a different literal in the same clause.  An example of such a path for
propositional calculus is the sequence $(\{p_1,p_2,p_3\}, (p_1, \neg p_1), \{\neg p_1, q_1,q_2\}, (q_1,\neg q_1),
\{\neg q_1, \neg r_1$, $\neg r_2\})$.  The sequence $(\{p_1,p_2,p_3\}, (p_1, \neg p_1), \{\neg p_1, q_1,q_2\}$, $(\neg p_1,p_1),
\{p_1, \neg r_1, \neg r_2\})$ is not an alternating path.

Why this particular definition of alternating path is chosen will become clear later, as its properties
are presented.

The {\em length} of an alternating path ($C_1, \dots, C_n$) is $n$, counting only the clauses.

The {\em relevance distance} $d_S(C_1, C_2)$ of $C_1$ and $C_2$ in $S$ is the length of the shortest alternating path in $S$ from $C_1$ to $C_2$.  This is a measure of how closely related to each other two clauses are.  If there is no alternating path
in $S$ from $C_1$ to $C_2$ then $d_S(C_1,C_2)$ is $\infty$.

If $T$ is a subset of set $S$ of clauses then $d_S(T,C)$ is $\min \{d_S(D,C) : D \in T\}$.  This is
called the {\em relevance distance} of $C$ from $T$.  Also, $R_{n,S}(T) =\{C \in S : d_S(T,C) \le n\}$.  This is frequently written as $R_n(T)$.  Clauses in $R_{n,S}(T)$ are said to be relevant {\em at distance n} from $T$.
If $d_S(T,C) < \infty$  then we say $C$ is {\em relevant} (for $S$ and $T$).

Two clauses are {\em alternating connected} or {\em relevance connected} in $S$ if there is an alternating path in $S$ between them.

Note that alternating connectedness is not transitive.  Example:  $C_1 = \neg P \vee Q, C_2 = \neg Q, C_3 = Q \vee \neg R$, and $S = \{C_1, C_2, C_3\}$.
$C_1$ and $C_2$ are alternating connected in $S$ as are $C_2$ and $C_3$, but $C_1$ and $C_3$ are not.

$Gr(S)$ is the set of ground instances of clauses in $S$.  If $S$ has no constant symbols then one such symbol is added for purposes of making $Gr(S)$ non-empty.

Two clauses $C$ and $D$ are {\em ground connected} in $S$ if they have ground instances $C'$ and $D'$ that are relevance connected in $Gr(S)$.

The {\em ground relevance distance} $d^g(C_1, C_2)$ of $C_1$ and $C_2$ is
$\min\{d_{Gr(S)}(C'_1,C'_2) : C'_1, C'_2 \mbox{ are ground}$ $\mbox{ instances of $C_1$ and $C_2$, respectively}\}$.

Sometimes $d^g$ can be larger than $d$, for example, consider $C_1$ and $C_3$ where
$C_1 = p(a), C_2 = \neg p(x) \vee p(f(x)), C_3 =\neg p(f(f(x))$, and $S = \{C_1, C_2, C_3\}$.

A set $S$ of clauses is {\em relevance connected} if between any two clauses $C$, $D$ in $S$ there is an alternating path.

\section{Properties of alternating path relevance}

\begin{theorem}
\label{relevance.connected.theorem}
If $S$ is a relevance connected set and $|S| = n$ then between any two clauses in $S$ there
is an alternating path of length at most $2n-2$.
\end{theorem}

\begin{proof}
This proof is essentially from Plaisted and Yahya \cite{DBLP:journals/ai/PlaistedY03}, with
a slightly different notation.  The idea is that if there is a relevance path between two clauses, then there is one in which each clause appears at most twice.  Also, the clauses at the ends of the path
clearly only need to appear once, or else there is a shorter path.
\end{proof}

\subsection{Minimal unsatisfiable sets of clauses}

This section relates alternating paths to unsatisfiability of sets of clauses.  This  section and section \ref{completeness.and.sos} are basically reviews of known results, with some new formalism.   First we have the following
result \cite{DBLP:journals/ai/PlaistedY03} :

\begin{theorem}
\label{relevance.connectedness}
If $S$ is a minimal unsatisfiable set of clauses then $S$ is relevance connected.
\end{theorem}

Using Theorem \ref{relevance.connected.theorem} above, we obtain the following:

\begin{theorem}
If $S$ is a set of clauses, $S'$ is a minimal unsatisfiable subset of $S$, and $|S'| = n$ then between any two clauses in $S'$ there is an alternating path in $S'$ of length at most $2n-2$. 
\end{theorem}

Now, short proofs imply small minimal unsatisfiable sets of clauses, which in turn implies that
if there is a short refutation from $S$ then there is a minimal unsatisfiable subset of $S$ in which any two clauses are connected by a short alternating path.

\begin{theorem}
If there is a resolution refutation of length $n$ from $S$ then there is a minimal
unsatisfiable subset $T'$ of $S$ in which any two clauses are connected by an alternating path
of length at most $2n-2$.
\end{theorem}

\begin{proof}
Suppose there is a resolution refutation of length $n$ from $S$.  Let $T$ be the set of input clauses (clauses in $S$) used in the refutation; then $|T| \le n$.  Also, $T$ is unsatisfiable
so it has a minimal unsatisfiable subset $T'$ with $|T'| \le n$.  Then any two clauses in $T'$
are connected by an alternating path of length at most $2n-2$.
\end{proof}

We are not aware of any such result that applies to other relevance measures.
If one wants to add a small proof restriction to other relevance measures, then one way to do this is to combine them with
alternating path relevance.

The following result is also easily shown, but without a bound on the length of the path:

\begin{theorem}
If $S$ is a minimal unsatisfiable set of clauses then any two clauses in $S$ are ground connected.  
\end{theorem}

\begin{proof}
The idea is that if $S$ is unsatisfiable then it has a finite unsatisfiable set of ground instances
which therefore has a relevance connected subset including at least one instance of each
clause in $S$.
\end{proof}

\subsection{Completeness and Set of Support}
\label{completeness.and.sos}
Using relevance, one can filter the potentially large set $S$ of input clauses (clauses in $S$) to obtain a smaller set
$S'$ of relevant clauses, and then one can search for a proof from $S'$ instead $S$.  This can
be done using the concept of a {\em set of support}.

\begin{Define}
If $S$ is an unsatisfiable set of clauses, then a {\em support set} for $S$ is a subset $S'$ of
$S$ such that any unsatisfiable subset of $S$ has non-empty intersection with $S'$.
\end{Define}

Such support sets are easily constructed from interpretations of the input clauses in many
cases.  In particular, if $I$ is an interpretation of the set $S$ of clauses, then $\{C \in S : I \not\models C\}$ is a set
of support for $S$.  If it is decidable whether $I \models C$ for clauses $C$, then such a set of support can be
effectively constructed.  This is true, for example, for finite models of $S$.

\begin{theorem}
If $S$ is unsatisfiable, $S'$ is a support set for $S$, there is a length $n$ refutation from $S$ and
$m \ge 2n-2$ then $R_{m,S}(S')$ is unsatisfiable.
\end{theorem}

This leads to the following theorem proving method $Rel_{m,S}(S')$:

\begin{center}
Choose $m$, compute $R_{m,S}(S')$, and test $R_{m,S}(S')$ for satisfiability.
\end{center}

If $|R_{m,S}(S')|$ is much smaller than $|S|$, then $Rel_{m,S}(S')$
can be much faster than looking for a refutation directly from $S$.
However, because one does not know which
$m$ to try, one can perform $Rel_{1,S}(S')$, $Rel_{2,S}(S')$, $Rel_{3,S}(S')$ et
cetera, interleaving the computations because even for a fixed $m$, $Rel_{m,S}(S')$ may not terminate.  This leads to the following theorem proving
method:

\begin{center}
for $i = 1$ step 1 until infinity do in parallel $Rel_{i,S}(S')$ od;
\end{center}

An example of a support set for $S$ mentioned above  is the set of clauses contradicting an interpretation $I$ of $S$.  This provides a way to use semantics ($I$) in theorem proving, which humans also commonly use.  Also, if $T$ is a satisfiable subset of $S$, then $S \setminus T$ is a support set for $S$.  For
example, if $T$ is a collection of axioms from some satisfiable theory such as number theory then $S \setminus T$ is a support set for $S$.  Typically one attempts to prove a theorem $R$ from
some collection $A$ of general, satisfiable axioms.  Then one converts $A \wedge \neg R$ to
clause form, obtaining set $S_1 \cup S_2$ of clauses where $S_1$ are the clauses coming from $A$ and $S_2$ are the clauses coming from $\neg R$.  Then $S_2$ is a support set for $S_1 \cup S_2$.  Support sets are often specified with a theorem.
\begin{Define}
If $S$ is unsatisfiable and $U$ is a support set for $S$, then the $U$ {\em support radius} for $S$, $U_S^{rad}$,
is the
minimal $m$ such that $R_{m,S}(U)$ is unsatisfiable.  If $S$ is satisfiable, then $m$ is $\infty$.
\end{Define}

The $U$ {\em support neighborhood} for $S$, $U_S^{nbrhd}$, is then $R_{m,S}(U)$.
A $U$ {\em support neighborhood clause} for $S$ is a clause in $U_S^{nbrhd}$ and a $U$ {\em support
neighborhood literal} for $S$ is a literal appearing in a $U$ support neighborhood clause for $S$.
The $U$ {\em support diameter} for $S$ is the maximum relevance distance between two clauses within the $U$ support neighborhood for $S$.

\section{Time Bound to Compute Relevance}

Relevance neighborhoods can be computed within a reasonable time bound, which makes it feasible to use relevance techniques in theorem proving applications involving large knowledge bases.  The computation methods presented here are new.  This
computation
requires finding all pairs of complementary unifiable literals in clauses of $S$, constructing a graph
from them, and then applying breadth-first search to compute the relevance distances from a support
set to all clauses
in $S$.

\subsection{Pairwise unification}
To compute the time for the unifications, let $||S||$ be the length of $S$ in characters when written out as a string of symbols, and
similarly $||C||$ and $||L||$ for clauses and literals, respectively.  Suppose $L_i$, $1 \le i \le n$ are the literals in $S$. Unification can be done in
linear time  \cite{PATERSON1978158}, so  testing all pairs of literals for
unifiability takes time proportional to $\Sigma_{1 \le i,j \le n} (||L_i|| + ||L_j||)$.  In practice, term
indexing  \cite{RaSeVo:01} permits this to be done much faster.  However, with some algebra, $\Sigma_{1 \le i,j \le n} (||L_i|| + ||L_j||) = 2n\Sigma_{1 \le i \le n}||L_i||$ which is quadratic in $||S||$.

\subsection{The graph $G_S$}
Let $G_S$ be a graph obtained from $S$ for purposes of computing relevance distances.  The
nodes $V$ of $G_S$ are triples $<L,C,in>$ and $<L,C,out>$ where $C \in S$ and $L \in C$.  There
are two kinds of edges in $G_S$:  Type 1 edges from $<L,C,out>$ to $<M,D,in>$ for all nodes
$<L,C,out>$ and $<M,D,in>$ in $V$ such that $L$ and $M$ are complementary unifiable. There
are also type 2 edges from nodes $<L,C,in>$ to $<M,C,out>$ where $L$ and $M$ are
distinct literals of $C$.  Type 1 edges encode the links between clauses in an alternating path and
type 2 edges encode switching from one literal of a clause to another literal in an alternating
path.  The number of edges can be quadratic in $||S||$.

A {\em path} in a graph is a sequence $(v_1, v_2, \dots, v_n)$ in which there is an edge from
$v_i$ to $v_{i+1}$ for all $i$.  The length of this path is $n$.
Then there is a direct correspondence between alternating paths in $S$ and paths in $G_S$.
Suppose $C_1, p_1, C_2, p_2, \dots,
C_{n-1}, p_{n-1}, C_n$ is an alternating path in $S$ and $p_i$ is the pair
$(L_i,M_{i+1})$ of literals.  The corresponding path in $G_S$ is
$<L_1,C_1,out>$, $<M_2, C_2, in>$, $<L_2, C_2,out>$, $<M_3,C_3,in>$, $\dots$,
$<M_n, C_n, in>$.  The length of the alternating path is $n$ but the length of the path
in $G_S$ is $2(n-1)$ if $n > 1$.

\subsection{Relevance neighborhoods}
Suppose one wants to find $R_k(U)$ where $U$ is a subset of $S$.  This can be done
using the well-known linear time breadth-first search algorithm  \cite{Cormen:2009:IAT:1614191}, which outputs the distances of all nodes from the starting node (and can be easily modified to have more than
one starting node).  First one sets the
distances of all nodes $<L,C,in>$ and $<L,C,out>$ for $C \in U$ to 1, and all other nodes have distances of
infinity. Then one applies breadth-first search which computes the length of all shortest paths from the nodes $<L,C,out>$ for $C \in U$ to all other nodes.  From this,
$R_k(U)$ is obtained as the clauses $C$ appearing in nodes $<L,C,in>$ whose distance is
less than or equal to $2(k-1)$. Because this graph has a size that is quadratic in $||S||$, the overall method is quadratic. However, to compute $R_k(U)$, it is only necessary to construct and search the
portion of the graph consisting of nodes at distance $2(k-1)$ or less, which can result in a
considerable savings of time, especially for very large knowledge bases and small $k$.

\subsection{The propositional case}
If $S$ is purely propositional then it is possible to find relevance neighborhoods much faster, as
follows:  The edges of type 1 are replaced by a smaller number of edges.
Suppose $C_1, C_2, \dots, C_m$ are all the clauses in $S$ containing a positive literal $P$ and
$D_1, D_2, \dots, D_p$ are all the clauses in $S$ containing $\neg P$.  Then the type 1 edges
from $<P,C_i,out>$ to $<\neg P,D_j, in>$ are replaced by edges from $<P,C_i,out>$ to
a new node $<P>$ and edges from $<P>$ to  $<\neg P,D_j, in>$.  Also, simlar edges are
added with the sign of $P$ reversed:  the type 1 edges
from $<\neg P,D_j,out>$ to $<P,C_i, in>$ are replaced by edges from $<\neg P,D_j,out>$ to
a new node $<\neg P>$ and edges from $<\neg P>$ to  $<P,C_i, in>$.  This means that
the path in $G_S$ corresponding to a path in $S$ becomes a little longer, but the number of
edges in $G_S$ is reduced from quadratic to linear, making the whole algorithm linear time.
Then an alternating path of length $n$ in $S$ corresponds to a path of length $3(n-1)$ in
$G_S$ if $n > 1$.  For small relevance distances, one need not construct all of $G_S$, as before.

Bounds similar to these but less precise were given earlier \cite{jepl:84,plaisted:80}.

\section{Branching Factor Argument}
Now we give a new evidence that relevance can help to find proofs faster.

Suppose each (first-order) clause has at most $k$ literals and each predicate appears with a given sign in at most $b$ clauses in
a set $S$ of clauses.  Suppose a clause $C$ is in an alternating path; how many clauses $D$ can appear
after it in alternating paths in $S$?  If $C$ is the first clause in the path then any one of its up to $k$ literals can
connect to at most $b$ other clauses, so there can be up to $bk$ clauses $D$ after $C$ in various alternating paths in $S$.  If $C$ is not the first clause in the path, then it cannot exit by the same literal it entered by, so the number of clauses $D$ that can appear after it in various alternating paths in $S$ is at most $b(k-1)$.   Thus there is a branching factor of at most $b(k-1)$ at each level except at the first level.

Suppose $S'$ is a support set for $S$. Then there can be $|S'|$ clauses $C_1$ that are the first
clauses in some alternating path starting in $S'$.  There can be up to $bk$ clauses $C_2$ after $C_1$ in alternating paths in $S$ and for
each clause $C_i$ for $i > 1$ there can be at most $b(k-1)$ clauses after it in various alternating paths in $S$.  So
in paths  $C_1, C_2, \dots, C_n$ starting in $S'$ with various clauses $C_1, \dots, C_n$, the number of clauses $C_n$ in all such paths is at most
$|S'|b^{n-1}k(k-1)^{n-2}$.  All clauses in $R_n(S')$ must appear in some such path.  The total number of clauses appearing in such paths is then bounded
by $|S'|\Sigma_{1 \le i \le n} b^{i-1}k(k-1)^{i-2}$.  If $bk > 1$ this is bounded by
$2|S'|b^{n-1}k(k-1)^{n-2}$.  If this quantity is muich smaller than $|S|$ and $R_n(S')$ is
unsatisfiable then the effort to find a proof from $R_n(S')$ can be much less than the effort to find
a proof from all of $S$.

The alternating path approach is intended for proofs with small relevance bounds, so that the exponent $n$ should be small.   Also, if clause splitting is used to break unifications, then the effective value of $b$ may be reduced.  The value $k$ is typically small in first-order clause sets.

\section{Experimental evidence with some knowledge bases}

There is another evidence that relevance can help to reduce the size of the clause set that one
must consider.  This is based on an implementation of relevance \cite{jepl:84} in which a first-order situation calculus knowledge base KB1 of
about 200 clauses and a first-order knowledge base KB2 of about 3000 clauses expressing a map of a portion of the United States were considered and relevance methods were
applied.  For these examples, $R_n(S')$ was computed for $S'$ representing various
queries and with increasing $n$ until a non-empty set of relevant clauses was found.  This approach
used additional pruning techniques to reduce the size of the relevant set such as a purity
test in which clauses were deleted from $R_n(S')$ if they had literals that did not complement
unify with any other literals in $R_n(S')$; additional details can be found in the paper.

\begin{tabular}{llllll}
Know.& Query & Dist. &No. of& No. & Note\\
base & &bound & clauses& needed\\
&&&found&for proof\\\\
1 & 1 & 5 & 10 & 7\\
1 & 2 & 6 & 24 & 18\\
2 & 3 & 4 & 22 & 6\\
2 & 3 & 4 &  7 & 6 & Different\\
&&&&& strategy\\
2 & 4 & 3 & 4 & 4 & \\
2 & 5 & 3 & 4 & 4 & \\
\end{tabular}

For query 5, the method instantiated the query itself and found four instances of the
query, all needed for a refutation.  In all cases a small set of unsatisfiable clauses was found.
These results were pubished previously \cite{jepl:84} but are not widely known.
\section{Large Knowledge Bases}

For some large theories, the actual refutations tend to be small; this also suggests that
relevance techniques can be helpful.  The following quote \cite{9f90b05ba36c4befb1996afc86ab124d} concerns the SUMO knowledge
base:

\begin{quotation}
"The Suggested Upper Merged Ontology (SUMO) has provided the TPTP problem library with problems that have large numbers of axioms, of which typically only a few are needed to prove any given conjecture."
\end{quotation}
However, in this case, one is frequently testing to see if the axioms themselves are satisfiable.
Because relevance techniques presented depend on knowing that a large subset of the axioms is satisfiable,
in order to use relevance one would have to find such a subset even without knowing that all
the axioms were satisfiable.
The following quotation  \cite{Ramachandran05first-orderizedresearchcyc:} concerns another large knowledge base:
\begin{quotation}
 "...the knowledge in Cyc's KB is common-sense knowledge. Common-sense knowledge is more often used in relatively shallow, 'needle in a haystack' types of proofs than in deep mathematics style proofs."
\end{quotation}
Concerning the Sine Qua Non approach \cite{DBLP:conf/cade/HoderV11}, the authors write: 
\begin{quotation}
"Problems of this kind usually come either from knowledge-base reasoning over large ontologies (such as SUMO and CYC) or from reasoning over large mathematical libraries (such as MIZAR). Solving these problems usually involves reasoning in theories that contain thousands to millions of axioms, of which only a few are going to be used in proofs we are looking for."
\end{quotation}
Of course, the knowledge base used for Watson \cite{Ferrucci:2012:ITW:2481742.2481743} was huge, but the proofs (if one can call them proofs) had to be found
quickly, so they had to be relatively small.

\section{Relationship to the Set of Support Strategy in Resolution}
Another evidence that relevance can help comes from the usefulness  in many cases of the set of support strategy \cite{woroca:65} for first-order theorem proving.  This strategy is included as one of the standard options in many first-order theorem provers.  This is because experience has shown that the set of support strategy often helps to find proofs faster.
Some experimental evidence that first-order theorem proving strategies using set of support techniques outperform
others for large theories has been obtained  \cite{Reif1998}.

Now we show formally that the set of support strategy restricts attention to relevant clauses, and in fact, uses the most relevant clauses first.  This result is new.  It is surprising that set of support should correspond to relevance defined in terms of alternating paths in this way, because the definition of alternating paths is non-intuitive.  Because the set of support strategy often helps to find proofs, this is evidence that relevance is also helpful for proof finding.

\begin{Define}
If $S$ is a set of first-order clauses then a {\em resolution sequence from $S$} is a sequence $C_1, C_2, C_3, \dots$
of clauses where each $C_i$ is either in $S$ or is a resolvent of $C_j$ and $C_k$ for $j,k < i$.  There is a {\em parent} function that for any $i$ returns
such a pair $(j,k)$, but mostly this will be left implicit.
\end{Define}

\begin{Define}
Suppose $S$ is a set of first-order clauses and $S'$ is a subset of $S$. Then
a clause $C_i$ in a resolution sequence $C_1, C_2, C_3, \dots$ from $S$ is $S'$-{\em supported} if it is either in $S'$ or at least one of its parents is $S'$-supported in the sequence.
The {\em set of support strategy} for $S$ with support set $S'$ is the set of resolution sequences $C_1, C_2, \dots, C_n, \dots$ from $S$ in
which each non-input clause $C_i$ is $S'$-{\em supported} in the resolution sequence.
\end{Define}

The set of support strategy is {\em complete} in the sense that if $S$ is an unsatisfiable set of first-order clauses and $S'$ is a support set for $S$ then there
is a set of support {\em refutation} from $S$, that is, a resolution sequence from $S$ according to the set of support strategy for $S$ and $S'$ in which
$C_n$ is the empty clause, denoting false.

\begin{Define}
Suppose $P$ is the alternating path $C_1, p_1, C_2, p_2, \dots,
C_{n-1}, p_{n-1}, C_n$ and $p_i = (L_i$, $M_{i+1})$ for $1 \le i < n$ and $n > 1$.  Let
$C'_i$ be $C_i \setminus \{L_i,M_i\}$ for $1 < i < n$.  Let $C'_1$ be $C_1 \setminus \{L_1\}$ and let $C'_n$ be
$C_n \setminus \{M_n\}$.  Then ${\cal P}^*$ is $C'_1 \cup C'_2 \cup \dots \cup C'_n$.  ${\cal P}^*$  is considered to be a clause
denoting the disjunction of its literals.  If $n = 1$ then
$(C_1)^* = C_1$.
\end{Define}

\begin{Define}
A collection ${\cal A}$ of alternating paths {\em covers} a first-order clause $C$ if for every literal $L \in C$
there is an alternating path $P \in {\cal A}$ and a literal $M \in P^*$ such that $L$ is an instance of
$M$.
\end{Define}

\begin{theorem}
\label{covering.resolvent}
Suppose ${\cal A}$ covers first-order clause $C$, and $D$ is a resolvent of $C$ and $C'$ for some
clause $C'$.  Let $L \in C$ and $L' \in C'$ be literals of resolution in the respective clauses.  Let
$P$ be a path in ${\cal A}$ and $L_P$ be a literal in $P^*$ such that $L$ is an instance of $L_P$.
Let $P'$ be a prefix of the path  $P$ such that the last clause in $P'$ contains $L_P$.  Let $\Gamma$ be
the path $P' \circ ((L_P,L'), C')$, that is, $\Gamma$ is $P'$ with $(L_P,L')$ and $C'$ added to the
end.  Then $\Gamma$ is an alternating path and ${\cal A} \cup \{\Gamma\}$ covers $D$.
\end{theorem}

\begin{proof}
The fact that $\Gamma$ is an alternating path follows directly from the definitions.  Let
${\cal A}'$ be ${\cal A} \cup \{\Gamma\}$.  Now, ${\cal A}'$ covers $D$; all literals in $D$ that
are instances of literals in $C'$ are covered because $C'$ is the last clause in $\Gamma$, and
the literal $L'$ which is possibly not in $\Gamma^*$ has been removed from $D$ by the resolution
operation.  The literals in $D$ that come from literals in $C$ are covered because the literals in
$C$ were already covered by $P'$ and the literals in $D$ have only been instantiated in the
resolution operation.
\end{proof}

Suppose $S$ is a set of first-order clauses and $S'$ is a support set for $S$.

 \begin{theorem}
\label{support.path.theorem}
Suppose $C_1, C_2, C_3, \dots C_n$ is a resolution sequence in the set of support
strategy for $S$ with support set $S'$.  Then for every literal $L$ in every derived (that is, non-input) clause $C_i$ there is an
alternating path $P$ from $S'$ of length at most $i$ such that $L$ is an instance of a literal in $P^*$; thus $L$ is an instance of
a literal in a clause $D$ in $S$ at relevance distance at most $i$.  Also, for every input clause
$C_i$ in the proof there is an alternating path $P$ from $S'$ of length at most $i$ ending
in $C_i$.
\end{theorem}

This theorem is saying that there is an alternating path of length 1 from $S'$ to $C_1$, an
alternating path of length at most 2 from $S'$ to $C_2$ if $C_2$ is an input clause, an alternating path of length at most 3
from $S'$ to $C_3$ if $C_3$ is an input clause, and so on.

\begin{proof}
This follows from Theorem \ref{covering.resolvent}. By induction on $i$, showing this for
$i$ = 1 and showing that if it's true for $i$ it's also true for $i+1$, there is a set
${\cal A}_i$ of alternating paths of length at most $i$ starting in $S'$ and covering $C_i$ if $C_i$ is a derived clause.
This implies all the conclusions of the theorem.
\end{proof}

The implication of Theorem \ref{support.path.theorem} for the set of support strategy is that all input clauses $C_i$ in the proof are relevant at distance at most $i$.

This shows that the set of support strategy restricts attention to relevant clauses, so that the
effectiveness of this strategy is evidence that relevance is helpful.  In fact, one can easily show that if $C$ is a clause in $S$
at relevance distance $n$ then $C$ appears in a set of support proof of length at most $2n-1$.  Thus the set of support strategy
uses exactly the relevant clauses in $S$.

\subsection{Limitations of the Set of Support Strategy}

The question now arises, if the set of support strategy is in some sense equivalent to relevance,
then why not just use it all the time and dispense with relevance altogether?  

First, there are some extensions to relevance that do not naturally incorporate into the set of support
strategy; these involve a purity test and the use of multiple sets of support.  These techniques
have been presented  \cite{plaisted:80,jepl:84} in a couple of early papers.

Also, for propositional clause sets, $DPLL$ with CDCL \cite{DBLP:series/faia/SilvaLM09} is generally much better than resolution,
so relevance techniques may have an advantage over resolution for such clause sets even with the set of support
strategy.

In addition, for a Horn set, that is, a set of first-order Horn clauses, with an all-negative set of support, the set of support strategy is the same as input resolution, which requires every resolution to have one parent that is an input clause.  Define the depth
of a proof so that the depth of a (trivial) proof of an input clause is zero, and if $C$ is proved
by resolving $C_1$ and $C_2$, then the depth of the proof of $C$ is one plus the maximum
depths of the proofs of $C_1$ and $C_2$. Then the depth of a proof from a Horn set with an all-negative set of support  is at least as large as the number of input clauses used, and could be larger if more than one instance of an input clause is used.  This implies that
it is possible to have a clause set with a small support radius but which requires a large number of input clauses and therefore a very deep set of support proof.  Thinking in terms of Prolog style
subgoal trees, the relevance distance corresponds roughly to the depth of the tree but the
length of a set of support proof corresponds to the number of nodes in the tree.  For such clause
sets, it may be better to use hyper-resolution, which basically resolves away all negative literals of a clause at once, because the proof depth corresponds only to the
depth of the tree.

As an example, consider the propositional clause set $\neg p, p \vee \neg q_1 \vee \neg q_2 \vee \neg q_3, q_1 \vee \neg r_1 \vee \neg r_2,
q_2 \vee \neg r_3 \vee \neg r_4, q_3 \vee \neg r_5 \vee \neg r_6$ together with the unit clauses
$r_1, r_2, r_3, r_4, r_5, r_6$.  Suppose $\neg p$ is the set of support.  Then the set of support strategy resolves $\neg p$ with $p \vee \neg q_1 \vee \neg q_2 \vee \neg q_3$ to produce
$\neg q_1 \vee \neg q_2 \vee \neg q_3$.  This then resolves with $q_1 \vee \neg r_1 \vee \neg r_2$
to produce $\neg r_1 \vee \neg r_2 \vee \neg q_2 \vee \neg q_3$.  Two more resolutions produce
$\neg r_1 \vee \neg r_2 \vee \dots \vee \neg r_6$ and six more resolutions with unit clauses
produce a contradiction, for a total of ten resolutions.  However, all clauses are within a relevance
distance of three of the set of support.  Also, hyper-resolution can find a refutation in fewer
levels of resolution.  First, the unit clauses hyper-resolve with $q_1 \vee \neg r_1 \vee \neg r_2,
q_2 \vee \neg r_3 \vee \neg r_4, q_3 \vee \neg r_5 \vee \neg r_6$ to produce the units
$q_1, q_2, q_3$ and these hyper-resolve with  $p \vee \neg q_1 \vee \neg q_2 \vee \neg q_3$ to
produce the unit $p$ which then resolves with $\neg p$ to produce a contradiction.  If the subgoal tree is larger the difference can be more dramatic; the number of levels of resolution by the set of
support strategy can be exponentially larger than the number of levels of hyper-resolution required.  If there are many clauses in the clause set, a deep proof can easily get lost in the huge search space that is generated.
The point is, although set of support restricts attention to relevant clauses, it does not always
process them in the most efficient way, so it may be better to separate relevance detection
from inference in some cases.  The set of support strategy intermingles a relevance restriction with resolution inference.

\section{DPLL}

This section presents an extension of alternating path relevance to the $DPLL$ method.  The $DPLL$ method for testing satisfiability of propositional clause sets is now of major importance in many
areas of computer science, especially with the extension to CDCL \cite{DBLP:series/faia/SilvaLM09}.  For example,
an open conjecture concerning Pythagorean triples was recently solved in this way \cite{10.1007/978-3-319-40970-2_15}.
Therefore any improvement in $DPLL$ is of major importance.
Relevance can decrease the work required for $DPLL$ if the number of $U$ support neighborhood
literals is small for a support set $U$.  This decrease in work requires a modification to $DPLL$.

The $DPLL$ method, without unit rules or CDCL, can be expressed this way, for $Q$ a set of propositional clauses:

\begin{tabular}{l}
$DPLL(Q) \leftarrow$\\
\ \ \ \ \ if the empty clause, representing false is in $Q$ then unsat\\
\ \ \ \ \ \ \ else\\
\ \ \ \ \ if $Q$ is empty then sat else\\
\ \ \ \ \ \ \ choose a literal $L$ that appears in $Q$;\\
\ \ \ \ \ \ \  if $DPLL(Q|L) =$ unsat then $DPLL(Q|\neg L)$ else sat\\
\end{tabular}

Here $Q|L$ is $Q$ with all clauses containing $L$ deleted and
$\neg L$ removed from all clauses containing $\neg L$.  Also, $Q|\neg L$ is defined similarly
with the signs of the literals reversed. We say that $L$ is
the {\em split literal} or the literal chosen for splitting, in this case.

$DPLL$ also has unit rules.  Essentially, if $Q$ contains a unit clause $L$ then $L$ is used to
simplify $Q$ by replacing all occurrences of $L$ by true and simplifying, and replacing all
occurrences of $\neg L$ by false and simplifying.   There is a similar rule if there is a unit
clause $\neg L$ in $Q$, with signs reversed.

\subsection{Relevant $DPLL$}

The material in this section is new.  Let $U$ be a support set for a set $S$ of propositional clauses.
Define the {\em relevance distance} $d(U,L)$ of a literal $L$ from $U$ to be the minimal $n$ such that $L$ or its
negation appears in a clause $C$ of relevance distance $n$ from $U$.

\begin{Define}
If $S$ is a set of clauses then ${\cal L}(S)$ is $\bigcup S$, that is, the set of literals in clauses of $S$.
\end{Define}

\begin{Define}
The {\em stepping sequence} for a set $S$ of propositional clauses and a support set $U$ for $S$
is a sequence $({\cal L}_1, {\cal L}_2, \dots, {\cal L}_n)$ where ${\cal L}_i$ is the set of literals
of $S$ at relevance distance $i$ from $U$, and $n$ is the maximum relevance distance from $U$
of any literal in $S$.  A {\em stepping remainder sequence} for $S$ and $U$ is a sequence
$({\cal L}'_1, {\cal L}'_2, \dots, {\cal L}'_n)$ such that for the stepping sequence
$({\cal L}_1, {\cal L}_2, \dots, {\cal L}_n)$ for $S$ and $U$, ${\cal L}'_i \subseteq {\cal L}_i$ for all $i$.  A
{\em leading literal} of a stepping remainder sequence $({\cal L}'_1, {\cal L}'_2, \dots, {\cal L}'_n)$ 
is an element of ${\cal L}'_i$ or the complement of such an element,
where $i$ is minimal such that ${\cal L}'_i$ is non-empty.  Also, $Step:T$ where
$T$ is a set of clauses and $Step =
({\cal L}'_1,\dots, \dots, {\cal L}'_n)$ denotes the stepping remainder
sequence $({\cal L}''_1, {\cal L}''_2, \dots, {\cal L}''_n)$ in which ${\cal L}''_j = {\cal L}'_j \cap {\cal L}(T)$
for all $j$.  If $Step$ is the stepping remainder
sequence $({\cal L}_1, {\cal L}_2, \dots, {\cal L}_n)$ then $|Step|$ is $\Sigma_i |{\cal L}_i|$.
\end{Define}

Note that support sets are easy to obtain for propositional clause sets.  For example, the set of all-positive clauses and the set of all-negative clauses are both support sets for all propositional clause sets.

The {\em relevant } $DPLL$ method $DPLL$-$Rel$ can be expressed this way, where $Q$ is a set
of propositional clauses and $StepR$ is a stepping remainder sequence for $Q$:

\begin{tabular}{l}
$DPLL$-$Rel(Q, StepR) \leftarrow$\\
1. \ \ \ \ \ if the empty clause, representing false is in $Q$\\
1. \ \ \ \ \ \ \ \ \ \ \ then unsat else\\
2. \ \ \ \ \ if $|StepR| = 0$ then sat else\\
3. \ \ \ \ \ \ \ choose $L$ to be a leading literal from $StepR$;\\
4. \ \ \ \ \ \ \ if $DPLL$-$Rel(Q|L,StepR:(Q|L)) =$ unsat\\
4. \ \ \ \ \ \ \ \ \ \ then $DPLL$-$Rel(Q|\neg L,StepR:(Q|\neg L))$\\
5. \ \ \ \ \ \ \ else sat\\
\end{tabular}

This procedure is called at the top level as $DPLL$-$Rel(S,Step)$ with $Step$ as the stepping sequence for $S$ and $U$, and
$U$ as a support set for $S$. If there is a choice of leading literals, then any DPLL heuristic can be used to choose among them.  Now, the recursive calls to $DPLL$-$Rel$ will have stepping sequences with fewer literals; that is,
$|StepR|$ will be
smaller with each level of recursion.  This enables proofs of properties of $DPLL$-$Rel$ by induction.  Also, for the recursive calls to $DPLL$-$Rel(Q,StepR)$, $StepR$ will always include all $U$ support neighborhood literals of $S$ that remain in $Q$, so that if $|StepR| = 0$
then $Q$ contains no $U$ support neighborhood literals of $S$.  If at this point $Q$ does not contain the empty clause, then one who
understands $DPLL$ can see that this means that
$DPLL$-$Rel$ has essentially found a model of the relevant clauses of $S$. By theorem \ref{relevance.connectedness},
$S$ is satisfiable.

These considerations justify returning "sat" if $StepR_i$ is empty for all $i$ in line 2 of $DPLL$-$Rel$.  Then a model of
the relevant clauses can be returned, without even exploring the remaining literals. 
However, this depends on $U$ being a valid support
set for $S$, that is, $S \setminus U$ is satisfiable.  If there is some doubt about this, then line 2 of $DPLL$-$Rel$ should be
replaced by the following:

\begin{tabular}{l}
2. \ \ \ \ \ if $|StepR| = 0$ then $DPLL(Q)$ else\\
\end{tabular}

\noindent
This means that ordinary $DPLL$ is called in this case to explore the remaining literals.

\begin{theorem}
If $S$ is an unsatisfiable set of propositional clauses, $Step$ is the stepping sequence for $S$ and a support set $U$ for $S$,
and $DPLL$-$Rel(S,Step)$ is called at the top level,
then the number of recursive calls to $DPLL$-$Rel(Q,StepR)$ for various stepping remainder sequences
$StepR$ is bounded by $2^k$ where $k$ is the
number of literals in the $U$ support neighborhood for $S$.
\end{theorem}

\begin{proof}
Let $m$ be the $U$ support radius for $S$.  If $StepR = ({\cal L}_1, {\cal L}_2, \dots, {\cal L}_n)$
is a stepping remainder sequence for $S$, let $StepR^j$ be $({\cal L}_1, {\cal L}_2, \dots, {\cal L}_j)$.
By induction on $|StepR|$ for the recursive calls, one shows that for every call to $DPLL$-$Rel(Q,StepR)$, the
set of clauses in $Q$ over the $U$ support neighborhood literals in $StepR^m$ is unsatisfiable.  The proof makes use of
the lemma that if $Q$ is unsatisfiable so are $Q|L$ and $Q|\neg L$ for any literal $L$.  If $|StepR^m| = 0$ then $R$ in
the recursive call must contain the empty clause.  Each recursive call to $DPLL$-$Rel$
reduces the value of $|StepR^m|$ by one or more, and there
are at most two recursive calls to $DPLL$-$Rel$, whence the $2^k$ bound follows.
\end{proof}

Note that $k$ can be much smaller than the set of all literals in relevant clauses in $S$.  This result is not trivial, because some clauses in $Q$ are deleted in $Q|L$ and $Q|\neg L$ and even
clauses that have literals that are not in $StepR^m$ can contribute to such deletions by a succession
of unit deletions in $DPLL$.  Therefore one needs to show that the clauses
over the literals in $StepR^m$ are still unsatisfiable in $Q|L$ and $Q|\neg L$.

This bound of $2^k$ can be much smaller than the worst case $2^n$ bound on ordinary $DPLL$ where $n$ is the
number of atoms (predicate symbols) in $S$.  It is not necessary to know the $U$ support
neighborhood  in order to apply this method.  This method can be extended to CDCL, as before. Of
course, in practice the number of calls is likely to be much smaller than $2^k$.
There are already heuristics for DPLL, but to our knowledge none of them decreases the worst case
time bound as this approach does.

If the $DPLL$ unit rules are used freely, units that are not in the $U$ support neighborhood may be processed
by the unit rules even if the
split literals are handled as in $DPLL$-$Rel$.  The choices to deal with this are to allow all units to be used in the unit rules, or to restrict the unit rules to only use units that are relevant in some sense.

\section{Other extensions of alternating path relevance}

\subsection{Splitting clauses}
Another idea that can help to make relevance more effective is to split a clause into several clauses
that together have the same ground instances.  For example, a clause $C[x]$ containing an
occurrence of a variable $x$ can be split into $C[f_1(\overline{y^1})], C[f_2(\overline{y^2})], \dots$, $C[f_n(\overline{y^n})]$ where
$f_1, f_2, \dots, f_n$ are all the function symbols appearing in the clause set and $\overline{y^i}$ are sequences of new
variables.  If the clause
set has no constant symbols, then one such symbol has to be allowed, in addition.  This idea can be
extended in the obvious way to clauses containing more than one variable.  This idea can be especially
hepful with general axioms such as the equality axioms.  For example, the axiom
$x = y \rightarrow y = x$ can join many clauses together, causing all clauses to be at small
relevance distances from one another.  This can hinder the application of relevance in systems
involving equality.  Splitting clauses that have literals unifying with the complements of literals in many other clauses
can help a lot in such cases.  Also, it can be helpful to choose which variable to split so that the largest number of complement unifications from literals in the clause is broken.
Spltting clauses was one of the techniques used earlier \cite{jepl:84}; the technique used there was to split a clause
$C[x]$ into two clauses $C_1[x]$ and $C_2[x]$ where
$C_1$ restricts $x$ to be instantiated to one of $f_1(\overline{y^1}) \cdots f_k(\overline{y^k})$ for some $k$ and $C_2$ restricts instantiation  of $x$ to one of $f_{k+1}(\overline{y^{k+1}}) \cdots f_n(\overline{y^n})$.

\subsection{A fine type theory}
Also, a very fine type theory can help with equality and with relevance in general.  Types can be incorporated into the unification algorithm, so that for example a variable of type "person"
would not unify with a variable of type "building."  This idea can increase the relevance distance
between clauses and make relevance more effective.

\subsection{Multiple sets of support}
If one has a finite collection $\{I_1, I_2, \dots, I_n\}$ of interpretations of a set $S$ of clauses, then let $T_i$ be
$\{C \in S : I_i \not\models C\}$.  Then the $T_i$ are sets of support for all such $i$.  Let $R^*_{n,S}$ be
$R_{n,S}(T_1) \cap R_{n,S}(T_2) \cap \cdots R_{n,S}(T_k)$.  Then if $S$ is unsatisfiable, it is easy to see that $R^*_{n,S}$ will also
be unsatisfiable for some $n$.  The method of Jefferson and Plaisted \cite{jepl:84} was basically to compute $R^*_{n,S}$ for
various $n$ and apply a purity test until a non-empty set was obtained.  Also, clause splitting was used.  These techniques
proved to be effective for the problems tried.

\section{Equality}
Equality can cause a problem for relevance if literals of the form $s = t$ or their negations appear
in many clauses.  It may be acceptable to use the equality axioms without any special modification for equality.  However, there are also other possibilities.  In general, there needs to be thought devoted to how to integrate equality and relevance, possibly using some kind of completion procedure \cite{bani:98} combined with relevance.

Ordered paramodulation is a theorem proving technique that is effective for clause form
resolution with equality.  It basically uses equations, replacing the large (complex) side of the equation by
the small (simple) side.  However, it is not compatible with the set of support strategy.
For example, consider this set:
$f(a,x)=x$, $f(y,b)=y$, $a \neq b$.  If $a \neq b$ is the set of support, then a refutation requires using the equations in the wrong direction for ordered paramodulation and even paramodulating  from variables.  This kind of paramodulation can be highly inefficient compared to ordered paramodulation.  This is
another evidence that just applying the set of support strategy is not always the best way to handle
relevance.

For equality and relevance there are at least several choices:  1.  Find the relevant clauses by some
method, then use a strategy such as ordered paramodulation and hyper-resolution to find the proof.  2.  Use Brand's modification method \cite{brand:75} on the set of clauses, find the relevant set of clauses, and
then apply some inference method to find the proof.  3.  Modify relevance to take into account equality, possibly by using a set of equations to rewrite things before computing relevance or by incorporating
E-unification \cite{handbook_unification} into the unification algorithm.

\section{Discussion}
The features of alternating path relevance have been reviewed, and some extensions including an extension to $DPLL$ have
been presented.  Graph based methods for computing this relevance measure have been presented.  Alternating path relevance has an unexpected relationship with the set of support strategy.  Some
previous successes with it as well as with the set of support strategy argue that this method can be effective.  A couple of
theoretical results indicating the effectiveness of alternating path relevance have also been presented. Evidence has been
given that several large knowledge bases frequently permit small proofs, suggesting that alternating path relevance and other relevance methods can
be effective for them.  An open problem for many of the relevance techniques is to give a theoretical justification for their
effectiveness.

\bibliography{relevance,sggs-inference-jar,survey13c,CADE15,largepaper2}

\begin{thebibliography}{USPV08}

\bibitem[BN98]{bani:98}
Franz Baader and Tobias Nipkow.
\newblock {\em Term Rewriting and All That}.
\newblock Cambridge University Press, Cambridge, England, 1998.

\bibitem[Bra75]{brand:75}
D.~Brand.
\newblock Proving theorems with the modification method.
\newblock {\em SIAM J. Comput.}, 4:412--430, 1975.

\bibitem[BS01]{handbook_unification}
Franz Baader and Wayne Snyder.
\newblock Unification theory.
\newblock In John~Alan Robinson and Andrei Voronkov, editors, {\em Handbook of
  Automated Reasoning}, pages 445--532. Elsevier and MIT Press, 2001.

\bibitem[CLRS09]{Cormen:2009:IAT:1614191}
Thomas~H. Cormen, Charles~E. Leiserson, Ronald~L. Rivest, and Clifford Stein.
\newblock {\em Breadth-first search}, pages 531--539.
\newblock The MIT Press, 3rd edition, 2009.

\bibitem[Fer12]{Ferrucci:2012:ITW:2481742.2481743}
D.~A. Ferrucci.
\newblock Introduction to "this is watson".
\newblock {\em IBM J. Res. Dev.}, 56(3):235--249, May 2012.

\bibitem[HKM16]{10.1007/978-3-319-40970-2_15}
Marijn J.~H. Heule, Oliver Kullmann, and Victor~W. Marek.
\newblock Solving and verifying the boolean pythagorean triples problem via
  cube-and-conquer.
\newblock In Nadia Creignou and Daniel Le~Berre, editors, {\em Theory and
  Applications of Satisfiability Testing -- SAT 2016}, pages 228--245, Cham,
  2016. Springer International Publishing.

\bibitem[HV11]{DBLP:conf/cade/HoderV11}
Krystof Hoder and Andrei Voronkov.
\newblock Sine qua non for large theory reasoning.
\newblock In {\em {CADE}}, volume 6803 of {\em Lecture Notes in Computer
  Science}, pages 299--314. Springer, 2011.

\bibitem[JP84]{jepl:84}
S.~Jefferson and D.~Plaisted.
\newblock Implementation of an improved relevance criterion.
\newblock In {\em First Conference on Artificial Intelligence Applications},
  pages 476--482, 1984.

\bibitem[MP09]{DBLP:journals/japll/MengP09}
Jia Meng and Lawrence~C. Paulson.
\newblock Lightweight relevance filtering for machine-generated resolution
  problems.
\newblock {\em J. Applied Logic}, 7(1):41--57, 2009.

\bibitem[MQP06]{DBLP:journals/iandc/MengQP06}
Jia Meng, Claire Quigley, and Lawrence~C. Paulson.
\newblock Automation for interactive proof: First prototype.
\newblock {\em Inf. Comput.}, 204(10):1575--1596, 2006.

\bibitem[Pla80]{plaisted:80}
D.~Plaisted.
\newblock An efficient relevance criterion for mechanical theorem proving.
\newblock In {\em Proceedings of the First Annual National Conference on
  Artificial Intelligence}, pages 79--83, 1980.

\bibitem[PSST10]{9f90b05ba36c4befb1996afc86ab124d}
Adam Pease, Geoff Sutcliffe, Nick Siegel, and Steven Trac.
\newblock Large theory reasoning with sumo at casc.
\newblock {\em AI Communications}, 23(2-3):137--144, 2010.

\bibitem[Pud07]{DBLP:conf/cade/Pudlak07}
Petr Pudlak.
\newblock Semantic selection of premisses for automated theorem proving.
\newblock In Sutcliffe et~al. \cite{DBLP:conf/cade/2007esarlt}.

\bibitem[PW78]{PATERSON1978158}
M.S. Paterson and M.N. Wegman.
\newblock Linear unification.
\newblock {\em Journal of Computer and System Sciences}, 16(2):158 -- 167,
  1978.

\bibitem[PY03]{DBLP:journals/ai/PlaistedY03}
David~A. Plaisted and Adnan~H. Yahya.
\newblock A relevance restriction strategy for automated deduction.
\newblock {\em Artif. Intell.}, 144(1-2):59--93, 2003.

\bibitem[RPS09]{DBLP:conf/cade/RoedererPS09}
Alex Roederer, Yury Puzis, and Geoff Sutcliffe.
\newblock Divvy: An {ATP} meta-system based on axiom relevance ordering.
\newblock In {\em {CADE}}, volume 5663 of {\em Lecture Notes in Computer
  Science}, pages 157--162. Springer, 2009.

\bibitem[RRG05]{Ramachandran05first-orderizedresearchcyc:}
Deepak Ramachandran, Pace Reagan, and Keith Goolsbey.
\newblock First-orderized researchcyc: Expressivity and efficiency in a
  common-sense ontology.
\newblock In {\em In Papers from the AAAI Workshop on Contexts and Ontologies:
  Theory, Practice and Applications}, 2005.

\bibitem[RS98]{Reif1998}
Wolfgang Reif and Gerhard Schellhorn.
\newblock {\em Theorem Proving in Large Theories}, pages 225--241.
\newblock Springer Netherlands, Dordrecht, 1998.

\bibitem[RSV01]{RaSeVo:01}
I.V. Ramakrishnan, R.~Sekar, and A.~Voronkov.
\newblock Term indexing.
\newblock In A.~Robinson and A.~Voronkov, editors, {\em Handbook of Automated
  Reasoning}, volume~II, chapter~26, pages 1853--1964. Elsevier Science, 2001.

\bibitem[RV99]{10.1007/3-540-48660-7_26}
Alexandre Riazanov and Andrei Voronkov.
\newblock Vampire.
\newblock In {\em Automated Deduction --- CADE-16}, pages 292--296, Berlin,
  Heidelberg, 1999. Springer Berlin Heidelberg.

\bibitem[SLM09]{DBLP:series/faia/SilvaLM09}
Jo{\~{a}}o P.~Marques Silva, In{\^{e}}s Lynce, and Sharad Malik.
\newblock Conflict-driven clause learning {SAT} solvers.
\newblock In {\em Handbook of Satisfiability}, volume 185 of {\em Frontiers in
  Artificial Intelligence and Applications}, pages 131--153. {IOS} Press, 2009.

\bibitem[SP07]{DBLP:conf/cade/SutcliffeP07}
Geoff Sutcliffe and Yury Puzis.
\newblock {SRASS} - {A} semantic relevance axiom selection system.
\newblock In {\em {CADE}}, volume 4603 of {\em Lecture Notes in Computer
  Science}, pages 295--310. Springer, 2007.

\bibitem[SUS07]{DBLP:conf/cade/2007esarlt}
Geoff Sutcliffe, Josef Urban, and Stephan Schulz, editors.
\newblock {\em Proceedings of the {CADE-21} Workshop on Empirically Successful
  Automated Reasoning in Large Theories, Bremen, Germany, 17th July 2007},
  volume 257 of {\em {CEUR} Workshop Proceedings}. CEUR-WS.org, 2007.

\bibitem[Sut16]{Sut16}
G.~Sutcliffe.
\newblock {The CADE ATP System Competition - CASC}.
\newblock {\em AI Magazine}, 37(2):99--101, 2016.

\bibitem[Urb04]{DBLP:journals/jar/Urban04}
Josef Urban.
\newblock {MPTP} - motivation, implementation, first experiments.
\newblock {\em J. Autom. Reasoning}, 33(3-4):319--339, 2004.

\bibitem[Urb07]{DBLP:conf/cade/Urban07}
Josef Urban.
\newblock Malarea: a metasystem for automated reasoning in large theories.
\newblock In Sutcliffe et~al. \cite{DBLP:conf/cade/2007esarlt}.

\bibitem[USPV08]{DBLP:conf/cade/UrbanSPV08}
Josef Urban, Geoff Sutcliffe, Petr Pudl{\'{a}}k, and Jir{\'{\i}} Vyskocil.
\newblock Malarea {SG1-} machine learner for automated reasoning with semantic
  guidance.
\newblock In Alessandro Armando, Peter Baumgartner, and Gilles Dowek, editors,
  {\em Automated Reasoning, 4th International Joint Conference, {IJCAR} 2008,
  Sydney, Australia, August 12-15, 2008, Proceedings}, volume 5195 of {\em
  Lecture Notes in Computer Science}, pages 441--456. Springer, 2008.

\bibitem[WPN08]{DBLP:conf/tphol/WenzelPN08}
Makarius Wenzel, Lawrence~C. Paulson, and Tobias Nipkow.
\newblock The isabelle framework.
\newblock In Otmane~A{\"{\i}}t Mohamed, C{\'{e}}sar~A. Mu{\~{n}}oz, and
  Sofi{\`{e}}ne Tahar, editors, {\em Theorem Proving in Higher Order Logics,
  21st International Conference, TPHOLs 2008, Montreal, Canada, August 18-21,
  2008. Proceedings}, volume 5170 of {\em Lecture Notes in Computer Science},
  pages 33--38. Springer, 2008.

\bibitem[WRC65]{woroca:65}
L.~Wos, G.~Robinson, and D.~Carson.
\newblock Efficiency and completeness of the set of support strategy in theorem
  proving.
\newblock {\em Journal of the Association for Computing Machinery},
  12:536--541, 1965.

\end{thebibliography}
\bibliographystyle{plain}

\end{document}